\documentclass[conference]{IEEEtran}
\IEEEoverridecommandlockouts

\usepackage{cite}
\usepackage{amsmath,amssymb,amsfonts}
\usepackage{graphicx}
\usepackage{textcomp}
\usepackage{xcolor}
\usepackage{multirow}

\usepackage{bm}
\usepackage{enumerate}
\usepackage{threeparttable}
\usepackage{tabularx} 
\usepackage{float}
\usepackage{caption}
\usepackage{mathtools}
\usepackage[shortlabels]{enumitem}
\usepackage{pgfgantt}
\usepackage{tablefootnote}

\usepackage{graphicx}
\usepackage{placeins}
\usepackage[normalem]{ulem}
\usepackage{latexsym}
\usepackage{caption}
\usepackage{subcaption}
\usepackage{amsmath}
\usepackage{amssymb}
\usepackage{wrapfig}
\usepackage{tabularx}
\usepackage{lscape}
\usepackage{float}
\usepackage{xcolor,colortbl}
\usepackage{times}
\usepackage{stfloats}
\usepackage{multirow}
\usepackage{enumerate}
\usepackage{xspace}
\usepackage{pifont}
\usepackage{calc}
\usepackage{fancyhdr}
\usepackage[shortlabels]{enumitem}
\usepackage{graphicx}
\usepackage{varioref}
\usepackage{lipsum}
\usepackage{url}
\usepackage{tikz}
\usepackage{booktabs}
\usepackage{algorithm}
\usepackage{amsthm}
\usepackage{algcompatible}
\usepackage[noend]{algpseudocode}
\usepackage{multirow}
\usepackage{longtable}
\usepackage{graphicx}
\usepackage{verbatim}
\usepackage{cite}
\usepackage{textcomp}
\usepackage{bm}
\usepackage{threeparttable}

\newcommand{\algrule}[1][.2pt]{\par\vskip.5\baselineskip\hrule height #1\par\vskip.5\baselineskip}
{\bfseries}{\rmfamily}
\newtheorem{mycorollary}{Corollary}{\bfseries}{\rmfamily}
\newtheorem{mylemma}{Lemma}{\bfseries}{\rmfamily}

\newcommand{\as}{\ensuremath {\leftarrow}{\xspace}}

\newcommand{\respond}{\ensuremath {\texttt{Respond}{\xspace}}}
\newcommand{\client}{\ensuremath {\texttt{Client}{\xspace}}}

\newcommand{\cmark}{\ding{51}}%
\newcommand{\xmark}{\ding{55}}%

\newcommand{\slap}{\ensuremath {\texttt{SLAP}{\xspace}}}
\newcommand{\pol}{\ensuremath {\texttt{PoL}{\xspace}}}

\newcommand{\ND}{\ensuremath {\texttt{ND}{\xspace}}}
\newcommand{\AP}{\ensuremath {\texttt{AP}{\xspace}}}
\newcommand{\dac}{\ensuremath {\mathcal{DAC}{\xspace}}}
\newcommand{\gs}{\ensuremath {\mathcal{GS}{\xspace}}}
\newcommand{\tlp}{\ensuremath {\mathcal{TLP}{\xspace}}}
\newcommand{\dbp}{\ensuremath {\mathcal{DBP}{\xspace}}}
\newcommand{\asrandom}{\ensuremath{\stackrel{\$}{\leftarrow}}}
\newcommand{\PSD}{\ensuremath {\texttt{PSD}{\xspace}}}

\def\BibTeX{{\rm B\kern-.05em{\sc i\kern-.025em b}\kern-.08em
    T\kern-.1667em\lower.7ex\hbox{E}\kern-.125emX}}
\begin{document}

\title{SLAP: Secure Location-proof and Anonymous Privacy-preserving Spectrum Access
}

\author{\IEEEauthorblockN{Saleh Darzi, Attila A. Yavuz}
\IEEEauthorblockA{\textit{Department of Computer Science and Engineering}, 
\textit{University of South Florida},
Tampa, FL, USA}
(salehdarzi@usf.edu, attilaayavuz@usf.edu)
}

\maketitle

\begin{abstract}
The rapid advancements in wireless technology have significantly increased the demand for communication resources, leading to the development of Spectrum Access Systems (SAS). However, network regulations require disclosing sensitive user information, such as location coordinates and transmission details, raising critical privacy concerns. Moreover, as a database-driven architecture reliant on user-provided data, SAS necessitates robust location verification to counter identity and location spoofing attacks and remains a primary target for denial-of-service (DoS) attacks. Addressing these security challenges while adhering to regulatory requirements is essential. 

In this paper, we propose $\slap$, a novel framework that ensures location privacy and anonymity during spectrum queries, usage notifications, and location-proof acquisition. Our solution includes an adaptive dual-scenario location verification mechanism with architectural flexibility and a fallback option, along with a counter-DoS approach using time-lock puzzles. We prove the security of $\slap$ and demonstrate its advantages over existing solutions through comprehensive performance evaluations.

\end{abstract}

\begin{IEEEkeywords}
Spectrum Access Systems, Location Privacy, Anonymous Credentials, Location Proof, Counter-DoS
\end{IEEEkeywords}

\section{Introduction}
\label{sec:introduction}
Spectrum Access Systems (SAS) have become the de facto technology for dynamic spectrum allocation, enabling efficient sharing among primary (PU) and secondary users (SU) while ensuring regulatory compliance and interference management. A notable example is the Citizens Broadband Radio Service (CBRS) in the U.S., operating in the 3.5 GHz band for federal and satellite services \cite{agarwal2022survey}. However, SAS introduces significant privacy and security challenges due to its reliance on continuous reporting of user location and transmission details to geo-location databases, raising concerns about user anonymity and privacy \cite{grissa2021anonymous}. The location-based nature of SAS also makes it vulnerable to spoofing, location fraud, and falsified data, increasing the risk of unauthorized spectrum access \cite{nguyen2019spoofing}. Furthermore, its database-driven architecture leaves SAS and Cognitive Radio Networks (CRNs) susceptible to denial-of-service (DoS) attacks, which compromise spectrum availability and system efficiency \cite{jakimoski2008denial}. Despite various solutions targeting privacy protection, location verification, and DoS resistance, existing approaches remain isolated and fail to address these issues comprehensively. In the following, we review relevant efforts related to our work. 


\subsection{Related Work}
\label{subsec:relatedwork} 

{\em Location Privacy and Anonymous Spectrum Access in Database-Driven CRN}:  
Compliance with Federal Communications Commission (FCC) regulations in centralized SAS requires the disclosure of sensitive user information, including precise location coordinates, spectrum channel preferences, usage data, and transmission details, to query spectrum availability. This mandatory reporting raises serious privacy concerns, such as location privacy breaches, identity tracing, and the exposure of behavioral patterns.  
Existing location-privacy schemes often have significant limitations. Many focus solely on SUs, neglecting PUs, where their impact on spectrum information is most critical. Computationally or information-theoretically secure Private Information Retrieval (PIR) methods demand resource-intensive operations or involve extensive communication with multiple non-colluding databases, imposing high computational and communication overhead \cite{grissa2021anonymous, darzi2024privacy, xin2016privacy}. Approaches based on k-anonymity and pseudo-identifiers fail to provide provable security, offering only weak privacy guarantees unless an impractically large k value is used, which is infeasible for large-scale networks with numerous users \cite{zhang2015optimal,zhu2019lightweight}. Similarly, differential privacy-based methods degrade the accuracy of spectrum availability information \cite{ul2022differential}. 
These shortcomings highlight the need for efficient mechanisms that ensure robust security, full anonymity, and strong location privacy against all network entities without compromising system performance and user experience.

{\em Location Proof and Spoofing Attack Resistance in SAS}: 
SAS, viewed as location-based services reliant on real-time user data, depend on the accuracy and integrity of this information for efficient and fair spectrum allocation. However, adversaries can exploit this reliance by impersonating legitimate entities or falsifying location and usage data to manipulate spectrum allocation, leading to spectrum interference, operational disruptions, and economic losses. Existing works addressing location proofs in SAS often fail to comprehensively mitigate broader threats, including location spoofing, distance fraud, mafia attacks, and distance hijacking \cite{nguyen2019spoofing, zeng2014location}.  
Many solutions rely on impractical assumptions, such as the existence of dedicated location-proof servers \cite{li2015privacy}, the inherent honesty of some entities \cite{xin2016privacy}, or the availability of trusted infrastructure like WiFi or cellular access points in all locations. These assumptions are unrealistic, especially in rural or sparsely populated areas where such infrastructure may be absent, limiting the applicability of these methods. Also, most schemes do not safeguard location privacy and anonymity against access points or location servers, leaving a significant gap.  
Thus, there is a pressing need for a practical and robust location verification mechanism in SAS that ensures privacy, anonymity, and resilience against diverse attack scenarios while aligning with the operational constraints of real-world deployments.

{\em DoS countermeasures for SAS and CRN Services}: 
The proliferation of inexpensive devices (e.g., IoT) and the reliance of SAS on geo-location databases have significantly amplified the risk of DoS attacks \cite{jakimoski2008denial}. These attacks overwhelm systems with malicious requests, disrupting spectrum allocation and degrading performance, particularly during spectrum usage notifications and CRN service requests.  
Proposed solutions include intrusion detection systems (IDSs), blockchain, cryptographic techniques like client puzzles, and game-theory-based methods \cite{darzi2024counter}. While AI-based detection excels at identifying attacks, it primarily focuses on detection rather than prevention and requires extensive network-wide knowledge and access to sensitive user traffic—an impractical approach for real-time SAS countermeasures. Similarly, client-puzzle protocols face challenges such as distribution inefficiencies, parallelization vulnerabilities, and excessive overhead on servers and users, limiting their feasibility.  
There is an urgent need for efficient DoS countermeasures tailored to tasks like spectrum usage notifications and CRN service requests, ensuring resilience without imposing significant resource burdens.

\subsection{Our Contributions}\vspace{-1mm}
\label{subsec:contribution}

{\em We developed an efficient framework that innovatively synergies advanced cryptographic techniques to address the complex privacy and security challenges posed by regulatory requirements on SAS such as DoS and spoofing attacks. The proposed scheme, "Secure Location-Proof and Anonymous Privacy-Preserving Spectrum Access ($\slap$)", is designed to meet these challenges effectively.} Key desirable properties of the $\slap$ framework are outlined below: 

\vspace{1mm}
$\bullet$~\noindent{\em \uline{Location Privacy-Preserving and Anonymous Spectrum Access}:}  We enable anonymous queries to geo-location databases while ensuring location privacy and compliance with FCC regulations: {\em (i)} Key operations, including spectrum queries by SUs, database population with PU information, location proof requests from WiFi access points (APs) or nearby devices, and CRN service requests, are executed using delegatable attribute-based anonymous credentials \cite{mir2023practical}. These credentials are built on structure-preserving signatures on equivalence classes with updateable commitments, ensuring unlinkability and untraceability for robust location privacy.  
{\em (ii)} Certified attributes in anonymous credentials containing device-specific information along with location proofs, ensure full anonymity during authentication while significantly improving the quality and reliability of SAS and CRN services.
{\em (iii)} Our comparison demonstrates that $\slap$ achieves spectrum query with significantly lower end-to-end delay, outperforming existing schemes by a wide margin: $17\times$ faster than \cite{darzi2024privacy}, $63\times$ faster than \cite{grissa2019trustsas}, $22.6\times$ faster than \cite{grissa2019location}, and $5.9\times$ faster than \cite{xin2016privacy}. Additionally, $\slap$ requires communication with only a single database and reduces total communication overhead by two orders of magnitude compared to PIR-based schemes that rely on multiple non-colluding databases.


\vspace{1mm}
$\bullet$~\noindent{\em \uline{Adaptive Location Proofs and Attack Resistance for Spectrum Access}:} We propose an adaptive location verification algorithm with dual scenario support and architectural flexibility: 
{\em (i)} \textit{AP-Based Verification:} In areas with accessible APs, the algorithm utilizes signal strength to verify user proximity and generate location proofs using group signatures \cite{tessaro2023revisiting}, ensuring privacy and anonymity. Compared to existing solutions, $\slap$ achieves location proof $2\times$ faster than \cite{li2015privacy} and $4\times$ faster than \cite{xin2016privacy}, while uniquely offering location privacy and full anonymity even against the AP itself. 
{\em (ii)} \textit{Device-Based Verification:} In the absence of APs, the algorithm employs public key distance-bounding protocols \cite{kilincc2016efficient} and delegatable anonymous credentials, enabling users to prove proximity to nearby devices and obtain delegated credentials with certified location proofs as attributes, all while preserving privacy and anonymity. 
This dual-scenario approach enhances resilience by providing a fallback mechanism, mitigating single points of failure, and ensuring robust architectural adaptability. It guarantees strong location privacy, full anonymity, and security against spoofing attacks, regardless of infrastructure availability or environmental conditions.

\vspace{1mm}
$\bullet$~\noindent{\em \uline{Counter-DoS Mechanism for Spectrum Usage}:} We propose a proactive defense mechanism utilizing time-lock puzzles. These puzzles, designed to resist parallelization, are generated using encryption schemes and tailored to device-specific capabilities recorded in user credential attributes, imposing no database storage overhead unlike other counterparts \cite{darzi2024privacy}. 



\vspace{-1.5mm}
\section{Preliminaries}
\label{sec:prelim}
This section outlines the notations, cryptographic preliminaries, and foundational components of our framework. 

\textbf{Notations:}  $|x|$ and $\{0,1\}^k$ signify the bit length of a variable and $k$-bit binary value, respectively. $\oplus$ represents XOR operation. 
$\{x_i\}_{i=1}^{\ell}$ and $\xleftarrow{\$}\mathcal{S}$ denote $(x_1, x_2, ..., x_\ell)$ and random selection from the set $\mathcal{S}$, respectively. Let $\mathbb{G}_1, \mathbb{G}_2$, and $\mathbb{G}_T$ be prime-order groups with order $p$, and let $e: \mathbb{G}_1 \times \mathbb{G}_2 \rightarrow \mathbb{G}_T$ denote a bilinear map satisfying bilinearity and non-triviality. 
$\textbf{m}[i]$ refers to the $i$-th element of the vector $\textbf{m}$ and $h(.)$ denotes a cryptographically secure hash function. $sk$ and $pk$ are secret and public keys, respectively. \vspace{+1mm}


\textbf{Delegatable Anonymous Credentials ($\dac$):} We use an attribute-based $\dac$ \cite{mir2023practical} for anonymous authentication, built upon structure-preserving signatures on equivalence classes of updatable commitments (SPSEQ-UC). The main algorithms are outlined below; additional details are available in  \cite{mir2023practical}:
\begin{itemize}[leftmargin=*]
\setlength{\itemsep}{1pt} 
    \item[-] $(pp, sk_{RI}, pk_{RI}) \as Setup(1^\lambda, 1^t, 1^\eta)$: Given the security parameter $\lambda$, an upper bound $t$ for the set commitment scheme's maximum cardinality, and a length parameter $\eta > 1$, it produces the system's public parameters $pp$ along with a signing key $sk_{RI}$ and a public key $pk_{RI}$ for each level $i \in [\eta]$ associated with the root issuer (RI), where $pp$ is implicitly provided as input to all subsequent algorithms.
    \item[-] $(pk, sk) \as KeyGen(pp)$: Given $pp$, it generates the user's key pairs ($sk, pk$), where $pk$ is the initial pseudonym.
    \item[-] $(nym, aux) \as NymGen(pk)$: Given a user’s public key $pk$, the algorithm generates a pseudonym $nym$ and auxiliary information $aux$ (randomness) for its usage.
    \item[-] $CreateCred(L', A, sk_{RI}) \leftrightarrow$ $GetCred(pk_{RI}, sk_u, A)$ $\rightarrow (cred, (\overrightarrow{C},\overrightarrow{O}), dk_{L'})$: An interactive algorithm between RI and a user identified by $nym_u$. Given $pp$, the RI’s public key $pk_{RI}$ and attribute set $A$, RI generates a delegatable root credential for the user via the SPSEQ-UC signature. This credential is rooted at $pk_{IR}$ and created for a set commitment $C$ certifying the attribute set $A$. The user receives the credential $cred$, the opening information $O$, and a delegatable key $dk_{L'}$ for level $L'$.
    \item[-] ${IssueCred(pk_{RI}, dk_{L'}, sk_u, cred_u, A_l, L'') \leftrightarrow ReceiveCr}$ ${ed(pk_{RI}, sk_r, A_l) \rightarrow (cred_r, dk'_{L''})}$: This interactive algo- rithm involves a delegator ($nym_i$) and a delegatee ($nym_r$). The delegator uses inputs including $pp$, $pk_{RI}$, attribute set $A_l$, delegation key $dk_{L'}$, secret key $sk_i$, credential $cred_i$, and auxiliary information $aux_i$ to generate a new credential $cred_r$. The delegatee, using their secret key $sk_r$ and $pk_{RI}$, receives $cred_r$ with an extended attribute set $A' = (A, A_l)$ and a delegation level $L''$ satisfying $L'' \leq L'$. The new credential includes an updated delegation key $dk'_{L''}$, allowing further delegation if permitted.     
    %
    \item[-] ${CredProve(pk_{RI}, sk_{p}, nym_p, aux_p, cred_p, D) \leftrightarrow CredVe}$ ${rify(pk_{RI}, nym_p, D) \rightarrow \{0, 1\}}$: This interactive protocol allows a credential holder to anonymously prove ownership of their credential to a verifier. The prover, identified by pseudonym $nym_p$, uses their secret key $sk_p$, auxiliary information $aux_p$, and credential $cred_p$ to generate a proof validating $cred_p$ with respect to a disclosed attribute set $D$. The verifier, using the RI's public key $pk_{RI}$ and the prover's pseudonym $nym_p$, verifies the proof against the disclosed attributes. If the proof is valid, the verifier outputs $1$; otherwise, it outputs $0$.    
\end{itemize}


\textbf{Distance Bounding Protocol ($\dbp$):} 
A $\dbp$ verifies the physical proximity of two network entities by measuring message transmission times during a rapid challenge-response exchange. We utilize a public key-based $\dbp$ \cite{kilincc2016efficient} built on a one-pass authenticated key agreement (AKA) protocol using the nonce-Diffie Hellman scheme to establish a session key between the prover ($P$) and verifier ($V$). This is further combined with a symmetric $\dbp$ \cite{vaudenay2015private} operating on the session key. The $\dbp$ comprises the following algorithms: 
\begin{itemize}[leftmargin=*]
\setlength{\itemsep}{1pt} 
    \item[-] ${ss \gets AKA(sk, pk, pk')}$: $P$ and $V$ derive the session key $ss$ using their own key pair and the other's public key $pk'$.
    \item[-] ${\{0,1\} \gets Sym\dbp(ss, \textit{th})}$: An interactive algorithm between $P$ and $V$ to verify proximity, given a distance threshold $\textit{th}$ and session key $ss$. {\em (1)} \textit{Initialization Phase:} V selects $m \in \{0,1\}^{2n}$, sends it to $P$. $P$ computes $a = ss \oplus m$. {\em (2)} \textit{Rapid Bit Exchange Phase (Time-Critical):} $V$ sends challenges ($c_i \in \{0,1\}$) to $P$, who computes responses ($r_i = a_{2i+c_i-1}$) and returns them. $V$ measures round-trip times ($timer_i$) over $n$ rounds.  
    {\em (3)} \textit{Authentication Phase:} V verifies proximity using $a = ss \oplus m$, the round-trip times, the allowed delay, and the speed of light. It checks ${timer_i \leq 2\times \textit{th}}$ and $r_i = a_{2i+c_i-1}$. If the prover is within $\textit{th}$, the algorithm outputs $1$; otherwise, $0$.
    %
\end{itemize}




\textbf{Group Signature ($\gs$):} 
A $\gs$ allows a group member to anonymously sign a message on behalf of the group \cite{yavuz2023beyond}. We adopt a variant of the BBS group signature scheme \cite{tessaro2023revisiting}, characterized by short signatures, provable security, and high efficiency, comprising the following algorithms:
\begin{itemize}[leftmargin=*]
\setlength{\itemsep}{1pt} 
    \item[-] ${pp_G \as BBS.Setup(1^\lambda)}$: Given the security parameter $\lambda$, it runs group parameter generation algorithm $GGen(.)$ and outputs ${(p, \mathbb{G}_1, \mathbb{G}_2, \mathbb{G}_T, e(.))}$. Then, it obtains ${g_1 \asrandom \mathbb{G}^*_1}$, ${g_2 \asrandom \mathbb{G}^*_2}$, and ${\textbf{h}_1 \asrandom \mathbb{G}^{\ell}_1}$, and returns the public parameters ${pp_G \as (p, g_1, \textbf{h}_1, g_2, \mathbb{G}_1, \mathbb{G}_2, \mathbb{G}_T, e(.))}$.
    \item[-] ${(sk, GK) \as BBS.KeyGen(pp_G)}$: Given $pp_G$, it computes ${x \asrandom \mathbb{Z}_p, X_2 \as {g_2}^x}$, and outputs (${sk \as x, GK \as X_2}$).
    \item[-] ${\sigma \as BBS.Sign(sk=x, \textbf{m})}$: Given the message $\textbf{m}$ and the secret key $sk$, it outputs the group signature ${\sigma = (A, \overline{e})}$ by performing ${C \as g_1 \Pi_i \textbf{h}_1[i]^{\textbf{m}[i]}}$, ${\overline{e} \asrandom \mathcal{D}_{\overline{e}}}$, and ${A \as C^{\frac{1}{x+\overline{e}}}}$.
    \item[-] ${\{0,1\} \as BBS.Verify(GK, \textbf{m}, \sigma=(A,\overline{e}))}$: On input the message $\textbf{m}$, signature $\sigma$, and the group public key $GK$, it checks ${C \as g_1 \Pi_i \textbf{h}_1[i]^{\textbf{m}[i]}}$ and returns $1$ if ${e(A, {g_2}^{\overline{e}}.vk) = e(C, g_2)}$; otherwise, returns $0$.
    %
\end{itemize}


\vspace{+1mm}
\textbf{Time-Lock Puzzle ($\tlp$):} Introduced by Rivest et al. \cite{rivest1996time}, $\tlp$ encrypts messages decryptable only after a set time. We adopt the RSA-based TLP \cite{jerschow2010offline}, leveraging non-parallelizable repeated squaring. Unlike hash-based puzzles, this prevents acceleration via multiple machines. The TLP includes: 
\begin{itemize}[leftmargin=*]
\setlength{\itemsep}{1pt} 
    \item[-] ${\Pi \gets Puzzle.Gen(1^\lambda, \kappa)}$: Given the security parameter $\lambda$, this algorithm follows the same procedure as the RSA key generation \cite{rivest1996time}, resulting in a private key $d$ and its modular inverse ${e = d^{-1} \pmod {\phi(n)}}$. The difficulty $\kappa$ is set as the number of modular squarings required, determined by $\kappa = T \cdot S$, where $S$ is the squarings-per-second rate of a reference machine and $T$ is the desired solving time. The value ${r = 2^\kappa \pmod {\phi(n)}}$ is computed, followed by the public exponent ${\Tilde{e} = 2^\kappa + \phi(n) - r + e}$, where ${z = \phi(n) - r + e}$. The lower bits of $\Tilde{e}$ are composed of $z$, prefixed by a sequence of 0-bits and a leading 1-bit. Finally, the algorithm outputs the $pk \as {\Pi = (n, \Tilde{e})}$ and the $sk\as {\psi = (n, d)}$, with the public key efficiently represented as ${(n, \kappa, z)}$.     
    \item[-] ${\psi \as Puzzle.Sol(m, (n, \Tilde{e}))}$: Given ${\Pi = (n, \Tilde{e})}$ and the message ${0 < m < n}$ chosen by the puzzle solver, it produces ${c = c_1 . c_2 \pmod n}$ where ${c_1 = m^{2^\kappa} \pmod n}$ and ${c_2 = m^z \pmod n}$. Then, it sets the solution as ${\psi = (m, c)}$.
    \item[-] ${\{0, 1\} \gets Sol.Verify(d, \psi)}$: Using the secret key $d$ and solution $\psi$, verify the correctness of ${c^d \pmod n = m}$ and return $1$ if valid; otherwise, return $0$.
\end{itemize}


\section{Framework Architecture and Model}
\label{sec:systemmodel}

\textbf{System Model:} Our framework model comprises five key entities: 1) {\em Federal Communications Commission (FCC):} The central authority governing the SAS, responsible for establishing system parameters and enforcing regulatory compliance. 2) {\em Private Spectrum Databases (PSDs):} These encompass multiple geo-location spectrum databases \cite{agarwal2022survey, grissa2021anonymous} that provide real-time spectrum availability data. PSDs operate in adherence to FCC regulations, ensuring synchronization and consistency. 3) {\em Users:} This group includes both primary users (PUs) and secondary users (SUs) equipped with various devices (e.g., laptops, IoT, smartphones). PUs supply spectrum usage data to PSDs, while SUs query these databases for spectrum availability and CRN services. Additionally, a Nearby Device (ND) refers to any verified user within proximity. 4) {\em Servers:} These are diverse network service providers (e.g., CRN, web, cloud servers) that clients access for specific services. 5) {\em Access Points (APs):} Existing WiFi access points or cellular network towers in the area equipped with synchronized clocks.

\textbf{Threat Model and Security Objectives:} 
Our threat model addresses various cybersecurity attacks, focusing on privacy, anonymity, and location spoofing: 
{\em (i)} Users location privacy and real identity are under threat in all stages of spectrum access and CRN services due to FCC's mandated requirement for sharing detailed coordinates, transmission data, and spectrum information.  In this model, PSDs handle query responses, and CRN servers provide network services, operating as honest-but-curious entities—fulfilling their roles while attempting to infer users’ location, identity, and personal information. 
{\em (ii)} Users are required to provide location proofs but may act maliciously to exploit spectrum channels and services or fall victim to spoofing or compromise. Providing incorrect locations could enable access to occupied channels or unauthorized services. Potential attacks include relay attacks, distance fraud, mafia fraud, and distance hijacking \cite{zeng2014location}.  
{\em (iii)} During spectrum usage notifications, channel access, or CRN service requests, users may launch DoS attacks targeting PSDs or CRN servers.

Given the system and threat models, $\slap$ aims to achieve the following security objectives: \\
$\bullet$ \underline{\textit{Client Privacy and Anonymity:}} User location coordinates, device specifications, and personal identity remain confidential and anonymous during spectrum access, usage notifications, and CRN services, safeguarding against PSDs, CRN servers, and external attackers. \\
$\bullet$ \underline{\textit{Location Verification and Attack Resistance:}} Users are restricted from accessing spectrum data or CRN services outside their verified location. The system is resilient to distance fraud, mafia fraud, and distance hijacking, ensuring only legitimate users at authenticated locations can access services. \\
$\bullet$ \underline{\textit{Denial-of-Service Resistance:}} Spectrum channels and CRN services are safeguarded against DoS attacks, whether from users or external sources, ensuring uninterrupted and reliable service availability. 
\vspace{-1mm}
\section{The Proposed Scheme: $\slap$}
\label{sec:scheme}

\subsection{SLAP Framework Architecture and Initial Setup}
\label{subsec:initialsetup}
\vspace{-1.5mm}

Geolocation databases store frequency information and synchronize as mandated by the FCC \cite{agarwal2022survey}. APs within a region function as a group, each holding a pair of secret key $sk_{AP}$ and the group verification key $GK$, generated by the FCC using ${(sk, GK) \gets \textit{BBS.KeyGen}(param_G)}$. To estimate a user’s physical distance, an AP performs signal strength analysis and round-trip time (RTT) measurements. Using the received signal strength (RSS), RTT, and environmental parameters ($env_{params}$), the algorithm $\Delta \gets \textit{ProxVerif}(RSS, RTT, env_{params})$ computes and outputs the estimated physical distance of the user. 
The FCC acts as the root issuer for credentials in the system. For a set of attributes $A$ associated with a user's device (e.g., device ID, type), the FCC issues Level 1 root credentials to all registered users (PUs and SUs) using the algorithm $CreateCred(L', A, sk_{FCC})$. Each user, identified by the pseudonym $nym_u$, obtains their credential via $GetCred(pk_{FCC}, sk_u, A) \rightarrow (cred_u, (\overrightarrow{C}, \overrightarrow{O}), dk_{L'})$. The credential $cred_u$ consists of a set commitment $\overrightarrow{C}$ over attributes $A$, rooted in the FCC's public key $pk_{FCC}$, the corresponding opening information $\overrightarrow{O}$, and a delegation key $dk_{L'}$ enabling delegation up to level $L'$. With this credential, the user can renew their pseudonym $nym_u$ or delegate their credentials to another user by switching to a new public key and optionally extending the attribute set to $A'$. Demonstrating possession of the credential involves the user proving ownership of the secret key $sk_u$ and generating a randomized signature over the required attributes.

\vspace{-3mm}
\subsection{SLAP Framework Main Operations}
\label{subsec:mainoperations}\vspace{-1mm}
The flow of the $\slap$ framework, depicted in Fig. \ref{fig:mainoperation}, comprises three main phases outlined as follows:

\begin{figure}
  \includegraphics[scale=0.44]{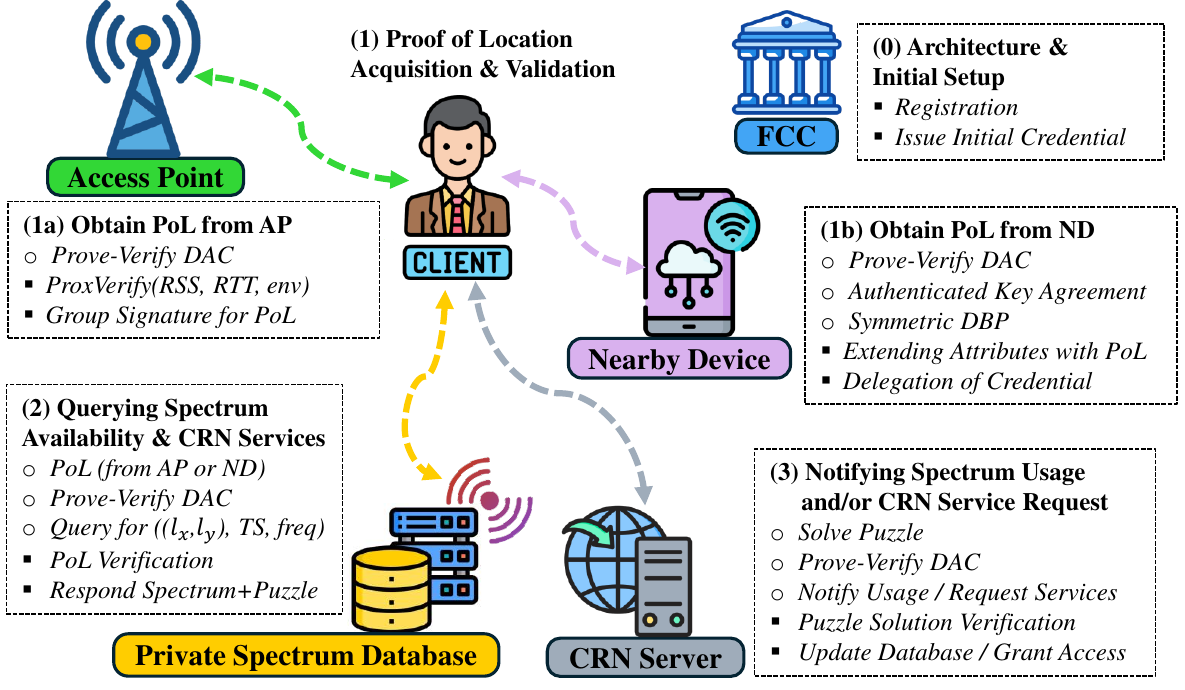} 
  \caption{$\slap$ Main Operations}
  \label{fig:mainoperation} 
\end{figure}

\subsubsection{\textbf{Proof of Location Acquisition and Validation}}  
In this phase, the user obtains a valid proof of location ($\pol$) for a specified time and geographic area, with the process tailored to two complementary scenarios: densely populated areas with robust infrastructure and rural regions with limited resources. 

{\em (i)} When an AP is within the user's proximity, the user requests a $\pol$ from the AP with the strongest signal, as outlined in Algorithm \ref{Alg:PoLAP}. Using anonymous credentials ($nym_c, cred_c$), the user specifies attributes $D$, timestamp ($TS$), and location coordinates ($l_x, l_y$), and verifies their credentials with the FCC's public key (Steps 1–3). Upon successful verification (Step 4), the AP evaluates proximity using signal strength and RTT measurements (Step 5). If proximity is validated (Step 6), the AP generates a group signature on the user's location, timestamp, and credentials and transmits it to the user (Steps 7–9). The user then verifies the group signature and accepts it as valid proof of location (Steps 10–11). 

\begin{algorithm}[ht!]
	\small
	\caption{${\Phi \as \pol.\AP(cred_c, (l_x,l_y), TS)}$}\label{Alg:PoLAP}
	\hspace{5pt}
	\begin{algorithmic}[1]
        \Statex \vspace{-1mm}$\textbf{Client}$:    
        \State Request $\pol$ from AP
        \State Set $D \as ((l_x, l_y), TS)$
        \State Perform $CredProve(pk_{FCC}, sk_{c}, nym_c, aux_c, cred_c, D)$ 
    \algrule
		\Statex \vspace{-1mm}$\textbf{Access Point}$:
            \If{$1 \as CredVerify(pk_{FCC}, nym_c, D)$}
            \State $\Delta_c \as \textit{ProxVerif}(RSS, RTT, env_{params})$
            \If{$(l_x, l_y) \in \Delta_c$}
            \State Set $\textbf{m}\as \{D, nym_c, Cred_c\}$
            \State $\sigma_{AP} \as BBS.Sign(sk_{AP}, \textbf{m})$
            \State Send $(\textbf{m},\sigma_{AP})$ to the Client.
            \EndIf
            \EndIf
	\algrule
	    \Statex \vspace{-1mm}$\textbf{Client}$:    
            \If{$1 \as BBS.Verify(vk, \textbf{m}, \sigma_{AP})$}
            \State \Return $\Phi \as (\sigma_{AP}, \textbf{m} = (l_x, l_y, TS, nym_c, Cred_c))$
            \EndIf
        \end{algorithmic}
\end{algorithm} 
\setlength{\textfloatsep}{0pt}

{\em (ii)} In sparsely populated areas lacking WiFi APs or cellular towers, this scenario provides a fallback for obtaining location proof and anonymous credentials from nearby devices (NDs), as detailed in Algorithm \ref{Alg:PoLND}. The client broadcasts a $\pol$ request to NDs for the current time and location (Step 1). Upon receiving responses, the client verifies their credentials with the ND using the FCC's public key (Steps 2–3). If valid, the ND establishes a secret session key $ss$ via an interactive authenticated key agreement and performs a symmetric $\dbp$ to verify the client's proximity within a threshold $\textit{th}$ (Steps 4–6). Once confirmed, the ND includes the client's location ($l_x, l_y$) and $TS$ in its attributes and anonymously delegates a credential to the client with limited delegation capabilities (Steps 7–10). Using the FCC's public key and their own secret key, the client receives the delegated credential and location proof, certified within the extended attributes (Steps 11–13).

\vspace{-1.5mm}
\begin{algorithm}[ht!]
	\small
	\caption{${(cred'_c, A') \as \pol.\ND(cred_c, (l_x, l_y), TS)}$}\label{Alg:PoLND}
	\hspace{5pt}
	\begin{algorithmic}[1]
        \Statex \vspace{-1mm}$\textbf{Client}$:    
        \State Request $\pol$ and a delegated credential $cred'_c$ from an ND
        \State Set $D \as ((l_x, l_y), TS)$
        \State Perform $CredProve(pk_{FCC}, sk_{c}, nym_c, aux_c, cred_c, D)$
    \algrule
	\Statex \vspace{-1mm}$\textbf{Nearby Device}$:
        \If{$1 \as CredVerify(pk_{FCC}, nym_c, D)$}
        \State Perform $ss \gets AKA(sk_{ND}, pk_{ND}, pk_c)$
        \If{$1 \as Sym\dbp(ss, \textit{th})$ and $(l_x, l_y) \in \textit{th}$}
        \State Set the new extended attributes as $A_l \as ((l_x, l_y), TS)$ 
        \State Set the new delegatable key as $dk'_{L''} := \perp$
        \State $IssueCred(pk_{FCC}, dk_{L'}, sk_{ND}, cred_{ND}, A_l, L'')$ 
        \State Send $(cred'_c, dk'_{L''})$ to the client.
        \EndIf
        \EndIf
	\algrule
	    \Statex \vspace{-1mm}$\textbf{Client}$:    
            \State $(cred'_c, dk'_{L''}) \as ReceiveCred(pk_{FCC}, sk_c, A_l)$
            \State Set $A' \as (A, \Phi)$ where $\Phi \as A_l$
            \State \Return ($cred'_c, A'$)
        \end{algorithmic}
\end{algorithm} \vspace{-3mm}
\setlength{\textfloatsep}{0pt}

\subsubsection{\textbf{Querying Spectrum Availability and CRN Services}} Algorithm \ref{Alg:query} details the process for querying spectrum availability, reporting spectrum usage, and accessing CRN services, primarily focusing on SUs as clients. The procedure for primary users PUs populating the database mirrors the process for querying $\PSD$s. Given the client’s location coordinates $(l_x, l_y)$ and the current timestamp $TS$, the process begins with obtaining a valid $\pol$, either from an AP or nearby devices.  
In areas with sufficient infrastructure, the client retrieves the proof from an AP using Algorithm \ref{Alg:PoLAP} (Step 1) and then proves their credentials to a $\PSD$ while querying for spectrum availability or CRN services (Steps 2–4). In poorly-infrastructured areas, the client obtains proof of location and delegated credentials from an ND using Algorithm \ref{Alg:PoLND} (Step 5). The delegated credential, containing the proof of location as an extended attribute, allows the client to anonymously prove their credentials to the $\PSD$ and submit queries for spectrum availability or CRN services (Steps 6–8). 
Notably, clients can precompute multiple credentials offline for future use, enhancing efficiency and flexibility.


Upon receiving a query, the $\PSD$ validates the credentials and proof of location. For AP-based location proofs, the $\PSD$ verifies the group signature, while for ND-based proofs, it checks the certified attributes, including the location proof, via the underlying signature verification (Steps 9–10). Based on the request for spectrum availability or CRN services, the $\PSD$ generates a puzzle linked to the target server’s public key (Steps 11, 14) and responds accordingly. 
While puzzle generation is included in the algorithm, $\PSD$s typically precompute puzzles with varying difficulty levels offline, similar to spectrum data. The difficulty is determined based on the risk of DoS attacks and the server's resource capacity to manage responses. Using the device details embedded in credential attributes, the $\PSD$ distributes the tailored puzzles accordingly. Notably, the online phase of $\slap$ only involves proving and verifying anonymous credentials during the query process, as location proof acquisition can be completed offline in advance.

\begin{algorithm}[ht!]
	\small
	\caption{$\slap$ Scheme}\label{Alg:query}
	\hspace{5pt}
	\begin{algorithmic}[1]
		\Statex \vspace{-1mm} $\textbf{Client}$: 
		\Statex $\underline{\rho_c \as Client.Query(cred_c, \Phi, (l_x,l_y), TS,~\textit{freq})}$: \vspace{+2mm}
            \Statex Give $(l_x, l_y)$ and $TS$, client request $\pol$: 
            \If{${\Phi \as \client.\pol.\AP(cred_c, (l_x,l_y), TS)}$}
            \State Set $D \as ((l_x, l_y), TS, \Phi)$
            \State Perform $CredProve(pk_{FCC}, sk_{c}, nym_c, aux_c, cred_c, D)$
            \State Query a $\PSD$ for $\rho_c \as ((l_x,l_y), TS,~\textit{freq})$     
            \EndIf
            \State \textbf{elseif} ${(cred'_c, \Phi) \as \client.\pol.\ND(cred_c, (l_x, l_y), TS)}$ \textbf{then}
            \State ~~~ Given $cred_c \as cred'_c$ with $A \as (A, \Phi)$ for $((l_x,l_y), TS)$
            \State ~~~ Perform $CredProve(pk_{FCC}, sk_{c}, nym_c, aux_c, cred_c, A')$
            \State ~~~ Query a $\PSD$ for $\rho_c \as ((l_x,l_y), TS,~\textit{freq})$
    \algrule
		\Statex \vspace{-1mm} $\textbf{Private Spectrum Database}$:
            \Statex $\underline{\rho_{\PSD} \as \PSD.\respond((l_x,l_y), TS,~\textit{freq})}$: \vspace{+2mm}
            \If{$1 \as CredVerify(pk_{FCC}, nym_c, D)$}
            \EndIf
            \If{${1 \as BBS.Verify(vk, \textbf{m}=(D,num_c, Cred_c), \sigma_{AP})}$}
            \State Set $\Pi \gets Puzzle.Gen(1^\lambda, \kappa)$ accordingly
            \State \Return $\rho_{\PSD} \as (\beta, \Pi)$ for $((l_x,l_y), TS)$
            \EndIf
            \State \textbf{elseif} $1 \as CredVerify(pk_{FCC}, nym_c, A')$ \textbf{then}
            \State ~~~ Set $\Pi \gets Puzzle.Gen(1^\lambda, \kappa)$ accordingly
            \State ~~~ \Return $\rho_{\PSD} \as (\beta, \Pi)$ for $((l_x,l_y), TS)$
    \algrule
        \Statex \vspace{-1mm} $\textbf{Client}$: Notifying spectrum usage to $\PSD$s or accessing services: 
	\For{Spectrum usage data or service requests as $m$}
        \State Given $\Pi \as (n, \Tilde{e})$ for the $\PSD$ or the target server
        \State Perform ${\psi \as Puzzle.Sol(m, (n, \Tilde{e}))}$
        \State Perform $CredProve(pk_{FCC}, sk_{c}, nym_c, aux_c, cred_c, D)$
        \State Sends $(m, \psi)$ to the $\PSD$ or the CRN server
        \EndFor
    \algrule
		\Statex $\textbf{Private Spectrum Database/CRN Server}$:
            \If{$1 \as CredVerify(pk_{FCC}, nym_c, D)$}
            \If{$1 \as Sol.Verify(m, c, d)$}            
            \State $\PSD$/CRN Server \Return $1$, and update DB or grant access
            \EndIf
            \EndIf
        \end{algorithmic}
\end{algorithm} \vspace{-2mm}
\setlength{\textfloatsep}{0pt}

\subsubsection{\textbf{Notifying Spectrum Usage and/or CRN Service Request}}  
To report spectrum usage data or access CRN services, users must solve the puzzle previously obtained, tied to the target server's public key. Given a message $m$, representing spectrum usage data or an access request, the user computes the puzzle solution via repeated squaring and submits it alongside proof of their anonymous credentials to the $\PSD$ or CRN server (Steps 16–20). Upon receiving the message and solution, the server validates the anonymous credentials (Step 21) and verifies the puzzle solution (Step 22). If both are verified, the $\PSD$ updates its database, or the CRN server grants access to the requested resources. Unlike other schemes, spectrum usage notifications also leverage anonymous credentials with attributes, improving frequency information quality while adhering to FCC coexistence requirements.

\section{Security Analysis}
\label{sec:security}
We provide security proofs addressing the threat model:
\vspace{-2mm}

\begin{mylemma}
\label{lem:anonymity&privacy}
$\slap$ ensures anonymous user authentication by leveraging the strong anonymity, soundness, and unforgeability properties of the ZKPoK and SPSEQ-UC signature schemes.
\end{mylemma} 
\begin{proof}\vspace{-1.5mm}
$\slap$ ensures robust anonymity, preventing any entity from tracing or inferring user identity or information beyond the required credentials during both issuance/delegation and presentation phases. Malicious verifiers cannot differentiate between users, and this strong anonymity is achieved without relying on a trusted setup. The framework's anonymity is grounded in the knowledge soundness of Zero-Knowledge Proof of Knowledge (ZKPoK), the Decisional Diffie-Hellman (DDH) assumption, and the SPSEQ-UC scheme \cite{fuchsbauer2019structure}, collectively ensuring origin-hiding, conversion-privacy, and derivation-privacy \cite{mir2023practical}. Origin-hiding guarantees indistinguishability of randomized signatures; derivation-privacy ensures extended commitment vectors remain indistinguishable; and conversion-privacy ensures new signatures generated with switched user keys are indistinguishable from fresh signatures. These privacy properties can be repeatedly applied in any order without compromising security.
\end{proof}

\vspace{-2mm}
\begin{mycorollary}\label{cor:unlikability}
$\slap$ provides location privacy for the spectrum access via the unlikability of the credentials formed on the signature and commitment pairs. 
\end{mycorollary}
\begin{proof}\vspace{-1.5mm}
The location privacy of $\slap$ is ensured by the unlinkability of signature-commitment pairs generated using the SPSEQ-UC scheme \cite{mir2023practical}. This unlinkability is achieved through signature re-randomization and user public key switching, enabling the repeated disclosure of the same commitment-signature pair without linkability. Provided no identifying attributes are included, newly generated signatures are indistinguishable from the originals. This property, formally proven secure under the group model, ensures that credential presentations remain unlinkable to verifiers. 
\end{proof}

\vspace{-1.5mm}
\begin{mylemma}
\label{lem:locationverif}
The $\slap$ framework ensures location verification of the users during spectrum access and queries via (i) the unforgeability of the group signatures and enhanced signal strength measurements; (ii) public key distance-bounding protocol and anonymous delegation of credentials.
\end{mylemma} 
\begin{proof}\vspace{-1.5mm}
In the first scenario, the risk of fraud against the AP is negligible due to robust security measures. Connection to the AP is secured using a broadcasted sequence number transmitted within a short time window (e.g., 100–500 ms), mitigating potential attacks, while proximity is validated through signal strength measurements. The AP’s group signature on the $cred$, $\pol$, and $TS$ verifies that the user is within the AP’s coverage area. The unforgeability of the $\gs$, grounded in the q-SDH assumption and supported by a tighter security proof in the algebraic group model, ensures the integrity of the location verification provided to the $\PSD$ \cite{tessaro2023revisiting}. Additionally, the location proof is non-transferable, as it is cryptographically bound to the current $TS$ and the user's verified credentials. 

In the ND scenario, location verification is ensured through the following mechanisms: (i) The security of the AKA protocol, based on the hardness of the Diffie-Hellman and discrete logarithm problems in the random oracle model \cite{kilincc2016efficient}. (ii) The negligible failure probability of the symmetric $\dbp$ \cite{vaudenay2015private}. Specifically, in the canonical OTDB scheme \cite{kilincc2016efficient}, with $m \in \{0,1\}^{2n}$ during initialization, the optimal probability for an adversary to correctly respond to all challenges is $(\frac{3}{4})^n$, providing strong resistance to distance fraud, mafia fraud, and distance hijacking \cite{kilincc2016efficient}. (iii) The unforgeability and anonymity of $\dac$ delegation. An adversary attempting to forge a new delegated credential with another user's certified $\pol$ must either forge the SPSEQ-UC scheme or compromise the NIZK proof scheme, both of which are provably secure \cite{mir2023practical}.
\end{proof}

\vspace{-3mm}
\begin{mycorollary}\label{cor:counterDoS}
$\slap$ offers a counter-DoS mechanism for spectrum access, usage notification, and obtaining CRN services via public-key time-lock puzzles.
\end{mycorollary}
\begin{proof}\vspace{-1.5mm}
The security of the TLP is directly grounded in Rivest's construction \cite{rivest1996time}, which relies on the hardness of the integer factorization problem and the computational properties of modular exponentiation with a power-of-two exponent. Specifically, deriving $c$ without performing $\kappa$ modular exponentiation operations during puzzle-solving (step 18 of Algorithm \ref{Alg:query}) is computationally infeasible for an adversary. Furthermore, reducing $\Tilde{e}$ to $e$ and computing $\phi(n)$ is provable as hard as factoring $n$ into its two large prime factors. To maintain security, the $\PSD$ must avoid disclosing multiple $\Tilde{e}$ values associated with the same key pair, as such disclosure would enable efficient factorization of $n$ and compromise the scheme. 
\end{proof}

\section{Performance Evaluation}
\label{sec:performance}

\begin{table*}[hbt!]
    \centering
    \renewcommand{\arraystretch}{1.2} 
    \setlength{\tabcolsep}{2pt}
    \begin{tabular}{|@{}c@{}|@{}c@{}|@{}c@{}|@{}c@{}|@{}c@{}|}\hline
    \textbf{Phase} & \textbf{Entity} & \textbf{Analytical Computational Cost} & \textbf{Empirical Cost} & \textbf{Communication Overhead} \\\hline
    \multirow{2}{*}{\textbf{PoL.AP}}& \textit{Client} & $((k+11)\mathbb{G}_1+3\mathbb{G}_2+\mathbb{G}_1^2)+\mathbb{G}_1^{|D|}++1P+3\mathbb{G}_T+(\sum_{i=1}^{|D|}(\mathbb{G}_1^{u_i}+\mathbb{G}_1))$& $20.17~\textit{ms}$ & $((k+8)|\mathbb{G}_1|+2|\mathbb{G}_2|+3|\mathbb{Z}_p|)$\\\cline{2-4} 
    & \textit{AP}& $2E^k+E^2+5E+\mathbb{G}_2^{|S|}+3\mathbb{G}_T+9\mathbb{G}_1+O(1)+\sum_{i=1}^{|D|}(\mathbb{G}_2^{|S-d_i|}+\mathbb{G}_2))$ & $61.26~\textit{ms}$ & $|TS|+|(l_x,l_y)| = 2008\textit{B}$ \\\hline\hline
    \multirow{3}{*}{\textbf{PoL.ND}} & \textit{Client} & $((k+3)\mathbb{G}_1+\mathbb{G}_2+\mathbb{G}_1^2)+E_M+H+rnd+O(1)+(\mathbb{G}_1^{|D|})+(\sum_{i=1}^{|D|}(\mathbb{G}_1^{u_i}+\mathbb{G}_1))$ & $31.75~\textit{ms}$ & $(3k+8)|\mathbb{G}_1|+4|\mathbb{G}_2|$\\ \cline{2-4}
    & \multirow{2}{*}{\textit{ND}}& $2E^k+E^2+5E+\mathbb{G}_2^{|S|}++ 3\mathbb{G}_1^2+2\mathbb{G}_1^n+\mathbb{G}_2^2)+(k+5)\mathbb{G}_1+\mathbb{G}_2$ & \multirow{2}{*}{$78.05~\textit{ms}$} & $+|TS|+|(l_x,l_y)|$\\ 
    & & $+\sum_{i=1}^{|D|}(\mathbb{G}_2^{|S-d_i|}+\mathbb{G}_2))+E_M+H+rnd+O(1)$ && $+(k+1)|\mathbb{Z}_p| = 1856\textit{B}$\\\hline\hline
    \textbf{Spectrum} & \textit{Client} &$((k+3)\mathbb{G}_1+\mathbb{G}_2+\mathbb{G}_1^2)+(\mathbb{G}_1^{|D|})+(\sum_{i=1}^{|D|}(\mathbb{G}_1^{u_i}+\mathbb{G}_1))+O(1)$ & $17.22~\textit{ms}$ & $(k+5)|\mathbb{G}_1|+|\mathbb{G}_2|+|\mathbb{Z}_p|+$ \\\cline{2-4} 
    \textbf{Query}& \textit{PSD} & ${2E^k+E^2+5E+\mathbb{G}_2^{|S|}+1P+3\mathbb{G}_T+2\mathbb{G}_2+8\mathbb{G}_1 + O(1)+\sum_{i=1}^{|D|}(\mathbb{G}_2^{|S-d_i|}+\mathbb{G}_2))}$ & $61.39~\textit{ms}$ & ${|TS|+|(l_x,l_y)|+|\beta| = 3080\textit{B}}$\\ \hline
    \textbf{Notify/}& \multirow{2}{*}{\textit{Client}} & $\kappa\times S_q + E_M+(k+3)\mathbb{G}_1+2\mathbb{G}_2+\mathbb{G}_1^2)$& $17.22~\textit{ms}$ & $(k+5)|\mathbb{G}_1|+|\mathbb{G}_2|+|\mathbb{Z}_p|$ \\
    \textbf{Service} & & $+(\mathbb{G}_1^{|D|})+(\sum_{i=1}^{|D|}(\mathbb{G}_1^{u_i}+\mathbb{G}_1)$& $+\kappa\times S_q$ & $|m|+|TS|+|\Pi|+|\psi|$\\\cline{2-4}
    \textbf{Request} & \textit{PSD}/ & {$2E^k+E^2+5E+\mathbb{G}_2^{|S|}+\mathbb{G}_2+\sum_{i=1}^{|D|}(\mathbb{G}_2^{|S-d_i|}+\mathbb{G}_2))$} & {$59.01~\textit{ms}$} & $=2304\textit{B}$\\\hline    
    \end{tabular}
    \begin{tablenotes}
       \item \textbf{Computations:} $\mathbb{G}_1$, $\mathbb{G}_2$, and $\mathbb{G}_T$ denote exponentiation in the respective groups. $E^k$ represents a $k$-pairing product, where $k=1$ corresponds to a single pairing operation; $P$ denotes pairing over the BN-256 curve. $E_M$ represents modular multiplication ($n=2048$). $rnd$ denotes random string selection, $H$ is a cryptographically secure hash function (SHA-256), and $S_q$ represents the repeated squaring time to solve a puzzle. $O(1)$ signifies signal transmission and internet communication time, typically in the microsecond range. Let $D = (d_i)_{i \in [k]}$ and $S = \bigcup_i d_i$ for all $i \in [k]$, where $k$ is the delegation level ($L=2$ in this scheme), and $(d_i, u_i)$ denotes disclosed and undisclosed attributes at level $i$. $\kappa$ represents puzzle difficulty.   \textbf{Communication:} Bits and bytes are denoted by $b$ and $B$, respectively. Group sizes are $|\mathbb{G}_1| = |\mathbb{Z}_p| = 256b$, $|\mathbb{G}_2| = 512b$, and $|\mathbb{G}_T| = 3072b$, with modular arithmetic over $n=2048$. Messages $|m| < 256B$, timestamps $|TS|$ are $8B$ (on a 64-bit Unix system), high-precision location coordinates are $16B$, and spectrum availability information $|\beta|$ (based on FCC raw data) is approximately $560~B$.
    \end{tablenotes}\vspace{-1mm}
    \caption{Computational Costs and Communication Overhead of $\slap$ Framework}\vspace{-6.5mm}
    \label{tab:Computational}
\end{table*}


\subsection{Metrics, Selection Rationale, and Configurations} \label{subsec:Configuration}
\noindent \textbf{Evaluation Metrics and Rationale:} 
We conduct analytical and empirical evaluations of the $\slap$ framework, assessing its computational costs and communication overhead across all phases and employed primitives, including $\dac$, $\gs$, $\dbp$, and $\tlp$. As no existing solutions offer a similarly comprehensive set of features, a direct performance comparison is not feasible. Instead, we provide a detailed performance analysis of $\slap$ across key metrics to evaluate its feasibility and practicality. Additionally, we present a qualitative and analytical comparison with selected schemes addressing subsets of these features in the context of spectrum query to SAS. The evaluation is structured as follows.

\noindent\textbf{Hardware, Software Libraries, and Parameters:} 
Our experiments were conducted on a desktop with an $11^{th}$ Gen Intel Core $\textit{i9-11900K} @ 3.50~GHz$, $64~\text{GiB}$ RAM, $1~\text{TB}$ SSD, running Ubuntu $22.04.4~\text{LTS}$. The implementation utilized libraries and tools such as \textit{DAC-from-EQS}\footnote{\url{https://github.com/mir-omid/DAC-from-EQS}},  
\textit{bbs-node reference}\footnote{\url{https://github.com/microsoft/bbs-node-reference/tree/main}}, \textit{time-lock-puzzle}\footnote{\url{https://github.com/pmuens/time-lock-puzzle}}, and \textit{OpenSSL}\footnote{\url{https://www.openssl.org/}}. These were used for implementing cryptographic primitives, including hash functions, modular arithmetic, exponentiation, and core $\slap$ components.  The setup included $\textit{SHA-256}$ for hashing, \cite{fuchsbauer2019structure} for set commitments, \cite{mir2023practical} for SPSEQ-UC, BN256 curve for binding and ECC, Schnorr-style ZKP with Damgard's technique for $\dac$ \cite{mir2023practical}, and NIZKs derived via the Fiat-Shamir heuristic, achieving approximately 100-bit security.

\vspace{-1mm}
\subsection{Experimental Results} \label{subsec:Performance}
\vspace{-1mm}
The analytical and empirical evaluation of cryptographic overhead, computational costs, and communication overhead for each phase of the $\slap$ framework is summarized in TABLE \ref{tab:Computational} and detailed below:

\noindent\textbf{Cryptographic Overhead:} 
To prove a credential, the user randomizes their $cred$ and $nym$ and employs a ZKPoK on the secret key $sk$ and randomness $aux$ to generate a new randomized $nym$ along with a subset of attributes $D$ using a set commitment scheme. Signature conversion, signature representation adjustment, and adaptation for a new set commitment take approximately $2 \textit{ms}$, $5 \textit{ms}$, and $13 \textit{ms}$, respectively.  
On commodity hardware, solving puzzles at difficulty levels $\kappa$ (number of squarings) set to $10^3$, $15 \times 10^3$, $50 \times 10^3$, $10^5$, and $10^6$ requires $3.9 \textit{ms}$, $56.31 \textit{ms}$, $194 \textit{ms}$, $784 \textit{ms}$, and $3.786 \textit{s}$, respectively. Verifying a puzzle solution, which involves RSA decryption (modular exponentiation), takes about $797 \mu s$. Group signing and verification are completed in $2.26 \textit{ms}$ and $3.17 \textit{ms}$, respectively, with batch verification reducing costs on the $\PSD$ side.  
In the employed $\dbp$, the AKA requires one ECC multiplication ($0.612 \textit{ms}$), one hashing ($0.35 \textit{ms}$), and random string selection ($0.045 \textit{ms}$). Rapid bit exchange occurs on a nanosecond scale, with a $\delta$ distance fraud probability corresponding to changes around $100 \textit{cm}$, negligible compared to other protocol aspects. For $ProxVerify()$, performed by the AP using signal strength and RTT techniques, the process is considered to take approximately $1$-$10 \textit{ms}$.


\noindent\textbf{Computational Costs:} 
{\em (i)} \textit{PoL.AP Phase:} The client proves anonymity and verifies the group signature, while the AP validates credentials, executes the $ProxVerify$ algorithm, and generates a group signature for the location proof.  
{\em (ii)} \textit{PoL.ND Phase:} This phase involves interactive protocols between two users, including credential proof and verification, key agreement, symmetric $\dbp$, and credential delegation/receipt with the location proof as a certified attribute.  
{\em (iii)} \textit{Spectrum Query Phase}: The user submits a valid location proof (from the AP or ND) when querying a $\PSD$ for spectrum availability at specific coordinates and timestamps. The $\PSD$ verifies credentials, checks the location proof, and provides spectrum information along with a public key puzzle tailored to the user’s attributes.  
{\em (iv)} \textit{Notifying Spectrum Usage or Requesting CRN Services Phase:} The user solves the $\PSD$ or CRN server's public key puzzle, proves credentials, and submits the solution. The $\PSD$/server verifies the solution before granting access to services or updating the database.

\noindent\textbf{Communication Overhead:} 
The communication complexity and data sizes for each phase are summarized in TABLE \ref{tab:Computational}. In our scheme, all attributes are assumed to have uniform size. The credential includes $|cred|+|sk|+|nym|$ within the set commitment and SPSEQ-UC schemes, maintaining a constant size independent of the number of attributes, calculated as ${4|\mathbb{G}_1|+|\mathbb{G}_2|+|\mathbb{Z}_p|}$, resulting in a credential size of $1792~b$. The size of $\overrightarrow{C}$ corresponds to the delegation level ($L=2$), with communication complexity increasing linearly with the number of attributes and delegations. Using publicly available raw database data from the FCC\footnote{\url{https://enterpriseefiling.fcc.gov/dataentry/public/tv/lmsDatabase.html}}, we estimated each database block to contain approximately $560~\text{bytes}$ of information, supplemented with synthetic data for evaluation purposes.

\begin{table*}[ht]
    \centering
    \renewcommand{\arraystretch}{1.2} 
    \setlength{\tabcolsep}{2pt}
    \begin{tabular}{|{l}|{c}|{c}|{c}|{c}|{c}||{c}|{c}|{c}|{c}|{c}|}\hline
         \multirow{2}{*}{\textbf{Scheme}} & \multicolumn{5}{c||}{\textbf{Features}} & \multicolumn{4}{c|}{\textbf{Delay}} & \textbf{Total} \\ \cline{2-10}
         & \textbf{Setting} & \textbf{Loc.Privacy} & \textbf{Anonym} & \textbf{Loc.Verification} & \textbf{Counter-DoS} & \textbf{SU} & \textbf{PSD} & \textbf{E2E} & \textbf{PoL} & \textbf{Communication}  \\\hline\hline
         \textit{Troja et al}\cite{troja2014leveraging} & \textit{1-DB} &  \textit{Peer-to-Peer} & \boldsymbol{\cmark} & \boldsymbol{\xmark}& \boldsymbol{\xmark} & $1650$ \textit{ms} & $11760$ \textit{ms} & $13410$ \textit{ms} & \boldsymbol{\xmark} & $12$ \textit{MB} \\\hline
         \textit{Li et al}\cite{li2015privacy} & \textit{1-DB} & \textit{Pseudo-ID} & \boldsymbol{\xmark} & \textit{WiFi AP+Loc.Server} & \boldsymbol{\xmark} & \boldsymbol{\xmark} & \boldsymbol{\xmark} & \boldsymbol{\xmark} & $210$ \textit{ms} & \boldsymbol{\xmark} \\\hline
         \textit{Xin et al}\cite{xin2016privacy} &\textit{1-DB} & \textit{PIR} & \boldsymbol{\xmark} & \textit{WiFi AP+QRA} &\boldsymbol{\xmark} & $292.8$ \textit{ms} & $142.7$ \textit{ms} & $407.4$ \textit{ms} & $430.1$ \textit{ms} & 325\textit{KB}\\\hline
         \textit{LP-Chor}\cite{grissa2019location} & \textit{$\ell$-DB} & \textit{PIR} &\boldsymbol{\xmark} & \boldsymbol{\xmark} & \boldsymbol{\xmark} & $7.7$ \textit{ms} & $480$ \textit{ms} & $620$ \textit{ms} & \boldsymbol{\xmark} & $753$ \textit{KB}\\\hline
         \textit{LP-Goldberg}\cite{grissa2019location} & \textit{$\ell$-DB} & \textit{PIR} &\boldsymbol{\xmark} & \boldsymbol{\xmark} & \boldsymbol{\xmark} & $320$ \textit{ms} & $1210$ \textit{ms} & $1780$ \textit{ms} & \boldsymbol{\xmark} & $6$ \textit{MB}\\\hline
         \textit{RAID-LP-Chor}\cite{grissa2019location} & \textit{$\ell$-DB} & \textit{PIR} &\boldsymbol{\xmark} & \boldsymbol{\xmark} & \boldsymbol{\xmark} & 
         $0.4$ \textit{ms} & $22$ \textit{ms} & $210$ \textit{ms} & \boldsymbol{\xmark}& $125$ \textit{KB}\\\hline
         Zeng et al\cite{zeng2019efficient} & \textit{1-DB} & \textit{BS+ECC} & \textit{PseudoID} &\boldsymbol{\xmark} &\boldsymbol{\xmark}  & $87$ \textit{ms} & $27$ \textit{ms} & $135$ \textit{ms} & \boldsymbol{\xmark} & $1.24$ \textit{KB}\\ \hline
         \textit{TrustSAS}\cite{grissa2021anonymous} &\textit{$\ell$-DB} &\textit{PIR} & \textit{EPID} & \boldsymbol{\xmark} & \boldsymbol{\xmark} & 
         $329.4$ \textit{ms} & $324.6$ \textit{ms} & $4954$ \textit{ms} & \boldsymbol{\xmark} & $1.25$ \textit{MB} \\\hline
         \textit{PACDoSQ}\cite{darzi2024privacy} & \textit{$\ell$-DB} &\textit{PIR}& \textit{Tor} & \boldsymbol{\xmark} & \textit{HBP} & $28.1$ \textit{ms} & $199$ \textit{ms} & $1373.6$ \textit{ms} & \boldsymbol{\xmark} & $605.92$ \textit{KB}\\\hline
         \multirow{2}{*}{$\slap$} & \multirow{2}{*}{\textit{1-DB}} & \multirow{2}{*}{$\dac$} & \multirow{2}{*}{$\dac$} & \textit{WiFi AP+G.Sig} & \multirow{2}{*}{\textit{TLP}} &\multirow{2}{*}{$17.22$ \textit{ms}}&\multirow{2}{*}{$61.39$ \textit{ms}}&\multirow{2}{*}{$78.61$ \textit{ms}}&$107.17$ \textit{ms}& \multirow{2}{*}{$3.08$\textit{KB}}\\ 
         &&&&$\dbp$+$\dac$&&&&&$109.8$ \textit{ms} &\\\hline
    \end{tabular}
    \begin{tablenotes}
       \item \textbf{Libraries:} Virtual Machines running Ubuntu simulated PIR costs, using the \textit{percy++} library\footnote{\url{https://percy.sourceforge.net/}} for multi-server PIR, the \textit{Open Quantum-Safe} library\footnote{\url{https://openquantumsafe.org/}} for PQC primitives, and \textit{OpenSSL} for cryptographic operations and arithmetic. \textbf{Variables:} We consider six databases for multi-DB schemes with $|DB| = 560 \textit{MB}$ and 400 rows/columns as described in \cite{xin2016privacy}. Key terms include \textit{BS} (base station), \textit{HBP} (hash-based puzzles), \textit{G.Sig} (group signature), \textit{QRA} (quadratic residue assumption), and \textit{EPID} (enhanced privacy ID based on direct anonymous attestation). 
    \end{tablenotes}\vspace{-1mm}
    \caption{Qualitative and Analytical Comparison with Existing Location Privacy Schemes}\vspace{-6mm}
    \label{tab:QuantitativeComparison}
\end{table*}

\noindent\textbf{Comparison with SOTA:} 
We perform a qualitative and analytical comparison of the achieved features with other state-of-the-art location privacy schemes, as detailed in TABLE \ref{tab:QuantitativeComparison}. For a fair evaluation, we consider spectrum query costs from both the client’s and $\PSD$'s perspectives, system communication overhead, and end-to-end delay for retrieving a single block from the geo-location databases as a measure of scalability. As shown in TABLE \ref{tab:QuantitativeComparison}, our approach delivers all necessary features for secure, location-private, and anonymous spectrum access, while offering architecture-flexible and efficient location verification with the lowest end-to-end delay and minimal communication burden on the system.

\vspace{-2mm}
\section{Conclusion}
\label{sec:conclusion}
The increasing demand for communication resources has driven the development of SAS, but regulatory requirements for disclosing sensitive user data raise privacy and security concerns, including the need for robust location verification and resilience against DoS attacks. To address these challenges, we proposed $\slap$, a framework ensuring strong location privacy, full anonymity during spectrum queries, and adaptive dual-scenario location verification while integrating TLP-based counter-DoS mechanisms. We formally proved its security and demonstrated its efficiency and scalability through comprehensive evaluations. 
\vspace{-2mm}

\section*{Acknowledgment}
This work is supported by the National Science Foundation NSF-SNSF 2444615. 

\vspace{-4mm}


\bibliographystyle{IEEEtran}
\bibliography{SalehRef}


\end{document}